\documentclass[lettersize,journal]{IEEEtran}
\usepackage{amsmath} 
\usepackage{amssymb} 
\usepackage{MnSymbol}
\usepackage{extarrows}
\usepackage{enumerate} 
\usepackage{bm} 
\usepackage{booktabs}
\usepackage{array}
\usepackage{amsthm} 
\usepackage{graphicx} 
\usepackage{subfigure} 
\usepackage{multirow} 

\usepackage{cite}
\usepackage{amsmath,amsfonts}
\usepackage{algorithmic}
\usepackage{array}
\usepackage[caption=false,font=normalsize,labelfont=sf,textfont=sf]{subfig}
\usepackage{textcomp}
\usepackage{stfloats}
\usepackage{url}
\usepackage{verbatim}
\usepackage{xcolor}
\usepackage{graphicx}
\usepackage[justification=centering]{caption}

\newtheorem{theorem}{Theorem}
\newtheorem{definition}{Definition}

\newtheorem{example}{Example}
\newtheorem{remark}{Remark}

\newcommand{\E}{\mathbb{E}}

\usepackage{epstopdf}

\newcommand{\MA}{\mathcal{A}}
\newcommand{\MB}{\mathcal{B}}
\newcommand{\MU}{\mathcal{U}}

\newcommand{\RR}{\mathbb{R}}

\newcommand{\EE}{\mathbb{E}}
\newcommand{\PP}{\mathbb{P}}

\allowdisplaybreaks[4] 

\date{}
\def\BibTeX{{\rm B\kern-.05em{\sc i\kern-.025em b}\kern-.08em
		T\kern-.1667em\lower.7ex\hbox{E}\kern-.125emX}}
\usepackage{balance}
\begin{document}

	\title{Coded Water-Filling for Multi-User Interference Cancellation
	}
\author{
   \IEEEauthorblockN{Yuan Li\IEEEauthorrefmark{1},
                     Zicheng Ye\IEEEauthorrefmark{1},
                     Huazi Zhang\IEEEauthorrefmark{1},
                     Jun Wang\IEEEauthorrefmark{1},
                     Jianglei Ma\IEEEauthorrefmark{1} and
                     Wen Tong\IEEEauthorrefmark{1}
\\}
   \IEEEauthorblockA{\IEEEauthorrefmark{1}%
                     Huawei Technologies Co. Ltd.\\}
    Email: tongwen@huawei.com}
	
	\maketitle
	
	\begin{abstract}
In this paper, we study the system-level advantages provided by rateless coding, early termination and power allocation strategy for multiple users distributed across multiple cells. In a multi-cell scenario, the early termination of coded transmission not only reduces finite-length loss akin to the single-user scenario but also yields capacity enhancements due to the cancellation of interference across cells. We term this technique \emph{coded water-filling}, a concept that diverges from traditional water-filling by incorporating variable-length rateless coding and interference cancellation.

We formulate a series of analytical models to quantify the gains associated with coded water-filling in multi-user scenarios.  First, we analyze the capacity gains from interference cancellation in Additive White Gaussian Noise (AWGN) channels, which arises from the disparity in the number of bits transmitted by distinct users. Building upon this, we broaden our analysis to encompass fading channels to show the robustness of the interference cancellation algorithms. Finally, we address the power allocation problem analogous to the water-filling problem under a multi-user framework, proving that an elevation in the water-filling threshold facilitates overall system capacity enhancement. Our analysis reveals the capacity gains achievable through early termination and power allocation techniques in multi-user settings. These results show that coded water-filling is instrumental for further improving spectral efficiency in crowded spectrums.
		
	\end{abstract}

\section{Introduction}\label{section1}
Next-generation wireless communications present new research opportunities. While much of the recent research efforts are made to harness the millimeter wave and centimeter wave bands, the refarming of low-frequency bands also requires novel wireless technologies. These bands, typically occupied by legacy technology, will be released for new technology. A question naturally arises: ``can we utilize these bands more efficiently than the previous way''? In this paper, we will present positive results bolstered by theoretical analysis.

The two main characteristics of low-frequency wireless transmissions are (i) less path loss and thus broader coverage, and (ii) narrower bands and thus longer transmission duration. A striking issue caused by these two characteristics is the multi-user interference, because unintended signals in these bands last longer and propagate farther. To further improve spectrum efficiency in these bands, one must seek to reduce the interference caused by these signals.

This paper focuses on the scenario with a set of users distributed across multiple cells, and proposes to early terminate successfully decoded transmissions so that the interference to other users can be reduced. In order to implement early termination, more frequent and precise feedback mechanisms must be available. These features are not supported in legacy systems.

This paper theoretically analyzes the proposed coded water-filling techniques, which comprise rateless coding, early termination and power allocation. Specifically, for any user within a single cell, the presence of other-cell users introduces interference, leading to a degraded signal-to-interference-plus-noise ratio (SINR). Timely termination of transmissions can reduce this cross-cell interference, and thus enhance system capacity. In multi-user scenarios, coded water-filling algorithms can achieve additional capacity gain through proper power allocation. Before delving into our analysis, we first introduce the concepts of feedback and water-filling, essential for understanding the intricacies of our approach.

\subsection{Variable-Length Feedback Codes}
	
In a single-user scenario, the capacity loss associated with fixed-length codes is directly proportional to the square root of the code length \cite{Strass}. A finite-length analysis of the relationship between the maximum coding rate and error probability in a point-to-point communication model with a fixed code length is established in \cite{poly1}. The maximum number of information bits for a fixed-length code can be approximated by:
\begin{equation}
\log M^*(n,\epsilon) = nC-\sqrt{nV}Q^{-1}(\epsilon) + O(\log n),
\end{equation}
where $M^*(n,\epsilon)$ is the maximum number of codewords with a code length of $n$ and an error probability of $\epsilon$, $C$ represents the channel capacity, $V$ is the channel dispersion, and $Q(x)=\int_{x}^{\infty} \frac{1}{\sqrt{2 \pi}} e^{-\frac{t^2}{2}}dt$. Shannon demonstrated that while feedback does not augment the capacity of point-to-point memoryless channels \cite{shannon}, it can simplify the coding scheme, as noted in \cite{fast1} and \cite{fast2}. Variable-length feedback codes, however, can flexibly adjust the code length according to channel noise power, thereby minimizing the finite-length capacity loss. In contrast, fixed-length codes must maintain a codeword length that ensures successful decoding under most channel noise realizations. However, variable-length feedback codes typically require the noiseless feedback about the received message to derive the noise of the forward channel. This can be both demanding and impractical. Therefore, stop-feedback codes are introduced to reduce the amount of feedback to one bit, and hence received considerable attention due to their simpler implementation. In variable-length stop-feedback (VLSF) codes, the receiver signals an acknowledgment to the transmitter by sending a single-bit ACK once it has gathered sufficient information to decode the message. Despite the minimal feedback overhead required, \cite{poly1} and \cite{poly2} demonstrate that VLSF codes can mitigate the finite-length rate loss efficiently. The finite-length achievable coding rate for a VLSF code can be approximated by:
\begin{equation}
\log M_f^*(\ell,\epsilon) = \frac{\ell C}{(1-\epsilon)} + O( \log \ell),
\end{equation}
where $0<\epsilon<1$, and $M_f^*(\ell,\epsilon)$ is the maximum number of codewords with an average code length of $\ell$ and an error probability of $\epsilon$. It is evident that the capacity is enhanced by a factor of $\frac{1}{1-\epsilon}$, and the gap to capacity narrows down from $O(\sqrt{\ell})$ to $O(\log \ell)$. This underscores the efficacy of VLSF codes in optimizing capacity under finite-length constraints.

The findings obtained in the single-user scenario are extended to the multiple access channel (MAC). Under the discrete memoryless two-user MAC model, \cite{Tril1} provides the achievability bounds for the joint decoding performance of VLSF codes by considering the joint mutual information $i(X_1,X_2;Y)$ under joint decoding. This means that users stop transmission simultaneously once the receiver has the capability to jointly decode messages from both users. Moreover, the authors also present the achievability bounds under successive interference cancellation (SIC) decoding. Here, the receiver first decodes the message of one user, then cancels the inference from the already-decoded user before decoding the other message. Although joint decoding significantly outperforms SIC decoding in terms of performance, the complexity of the former grows exponentially as the number of users increases. The results in \cite{Tril1} are extended to the Gaussian MAC with power constraints \cite{gassian1}. The researchers developed techniques to obtain achievability and converse bounds for Gaussian channels, and provided the second-order asymptotic term for the Gaussian MAC. Inspired by concepts related to composite channels \cite{compound}, Trillingsgaard \emph{et. al.} studied discrete memoryless broadcast channels with stop feedback for two users \cite{Tril2} and for $K$ users \cite{Tril3} with a common message transmission.
Although users in \cite{Tril3} have distinct stopping times, the interference cancellation benefits of early termination are not manifested, as inter-user interference is not taken into account.
	
In an effort to minimize the number of decoding attempts, feedback is made available only at $L$ predetermined decoding times. In \cite{Will1}, Williamson \emph{et. al.} numerically optimized the value of the decoding time $L$, and employed punctured convolutional codes and a Viterbi decoding algorithm to provide a practical scheme. R. C. Yavas \emph{et. al.}. \cite{time1}  devised an integer programming procedure aimed at minimizing the upper bound of the average block length across all decoding times, subject to constraints imposed by the average error probability and the integer condition. Recently, in \cite{review}, R. C. Yavas \emph{et. al.} compiled an overview of the findings on VSLF codes, examining the achievability bounds and converse bounds of VSLF codes in point-to-point, multiple access, and random access communication scenarios.

\subsection{Water-filling}

	The analysis of Gaussian channel capacity was first undertaken by Shannon in 1948 \cite{shannon1}, a seminal work that laid the groundwork for modern information theory. Building upon this, the water-filling solution, a pivotal technique for maximizing the capacity of parallel Gaussian channels, was introduced in \cite{shannon2}. This approach has since been refined and adapted to accommodate the complexities of time-continuous Gaussian channels, as evidenced by the comprehensive studies in \cite{timeconti1,timeconti2,timeconti3,timeconti4,timeconti5}. The classical water-filling method was introduced in detail in \cite{water2}, this method ingeniously parallels the allocation of power across channels to the distribution of water in a series of containers. The water-filling technique has been successfully applied to a wide range of communication scenarios, as showcased in \cite{water3.1}, \cite{water3.2}, and \cite{water3.3}. Notably, its utility extends to optimizing objectives beyond capacity maximization, such as the minimization of bit-error rate (BER) across a set of parallel channels, as explored in \cite{water4}.

In \cite{constant1}, the authors studied a discrete-time fading channel model for a single-user scenario. They investigated the capacity of this fading channel under the constraint of average transmission power. The signal-to-noise ratio (SNR) fluctuations induced by random channel fading coefficients resemble those of parallel Gaussian channels. Recognizing the complexity involved in the computation of the optimal water-filling method, the authors of \cite{constant11} observed that by ``allocating zero power to channels that would receive zero power under precise water-filling and a constant power in the remaining subchannels'', the optimal solution can be approximated. An approximate water-filling method, incorporating constant power allocation, is proposed in \cite{constant2}. After deriving the lower bound of the dual optimization problem, the author proved that the approximate water-filling scheme is close to optimal. A similar result is shown in \cite{water15}, as the signal-to-noise ratio approaches infinity, the water-filling power allocation function converges pointwise to a function that allocates fixed power to all non-zero channel gain states. In \cite{constant3}, the water-filling problem is extended to the MAC scenario. The optimal strategy is again water-filling, albeit a condition that transmission privileges are granted to the user with the most favorable channel conditions. This adaptation ensures that the benefits of water-filling are maximized while mitigating the effects of interference among multiple users.
	
The exploration of water-filling problems has been extended to encompass the Multiple-Input Multiple-Output (MIMO) scenario. In the Single-Input Single-Output (SISO) systems, traditional water-filling techniques offer a closed-form solution to the optimal transmission problem. In contrast, the MIMO variant of the water-filling problem often presents a more intricate challenge, leading to nonconvex problem formulations that are notably harder to solve. \cite{water5} broadens the scope of water-filling research by considering various optimization criteria, including the mean square error (MSE), SNR, and BER, demonstrating the versatility of the water-filling approach across different design metrics. \cite{water6.1} further demonstrates a system that aims at minimizing the MSE matrix determinant in MIMO configurations. The iterative water-filling algorithm, as presented in \cite{water13}, has become a prevalent algorithm for tackling power allocation optimization in MIMO broadcast channels \cite{water7} \cite{water8}. For a sufficiently large number of users, the channel capacity achieved with a water-filling scheme converges to a constant value, irrespective of the inherent randomness of the MIMO channel \cite{water9}. Given the widespread utility of the water-filling method, the design and optimization of water-filling algorithms have become important research topics. Works such as \cite{water10}, \cite{water11}, and \cite{water12} continue to enhance the efficiency and performance of water-filling techniques.
	
Water-filling is also extended to a variety of communication scenarios, showcasing its adaptability and significance across diverse contexts. In \cite{water20}, the authors tackle the water-filling problem under a fading channel model with multiple users in the presence of crosstalk. They propose a modified iterative water-filling algorithm that incorporates an updated crosstalk term. The work of \cite{water22} studies a scenario where a specific user communicates through dual base stations. The authors derive a closed-form solution to the optimal power allocation problem in this complex setting, upon which they propose a novel joint power allocation scheme. \cite{water23} examines a MAC with time-varying Gaussian fading for multiple users, where the channel remains static within each unit of time. Departing from the conventional constraint of constant total power, the optimization problem is redefined to focus on a scheduling strategy that minimizes energy consumption. Building upon this novel constraint, the authors propose an iterative dynamic water-filling algorithm, which dynamically adjusts power allocation to optimize energy efficiency in time-varying channels. The water-filling method can be combined with the game theory, as demonstrated in \cite{water27}. Here, the water-filling power allocation problem is characterized as a game, with power allocation corresponding to each user's strategic control of power to maximize their individual code rates. The authors prove that, within this game model, the water-filling game reaches a unique Nash equilibrium.
	
\subsection{Contributions}
In this paper, we study the system-level advantages provided by early termination and power allocation strategy for multiple users distributed across multiple cells. Our innovation lies in the fact that, unlike traditional multi-user variable-length codes where all users stop simultaneously, we consider varying stopping times for different users. Once a user stops transmission, interference cancellation can further enhance the SINR for other users, thereby increasing the system throughput. Regarding power allocation, we find that within a multi-user framework, each user should adopt a more assertive power allocation strategy. This entails elevating the water-filling threshold relative to a single-user environment, therefore minimizing the interference to other users by transmissions conducted at low SNR levels.

Our main contributions and findings are summarized as follows:
\begin{enumerate}
		\item  Given that users are distributed across distinct cells, implementing joint decoding or SIC is not feasible. Consequently, users within every cell treat messages originating from other cells as interference. We prove the advantages of variable-length feedback codes within this context. These codes not only mitigate the finite-length loss, scaling it down from $O(\sqrt{\ell})$ to $O(\log \ell)$, as observed in single-user scenarios, but also introduce additional capacity enhancements. The early termination of transmission coupled with interference cancellation contributes to boosting the overall system throughput, underscoring the significant gains on the system level.
		\item We develop a multi-user queuing model. Within this scenario, during periods when all cells are experiencing congestion, interference reaches its peak level. Under such circumstances, the early termination of transmissions does not facilitate interference cancellation, but only obtains the feedback gain. Conversely, when the queue exhibits sufficient sparsity, both interference cancellation gain and feedback gain can be obtained. Our findings highlight that throughput enhancements is maximized when user density falls within a moderate range. This reveals the interplay between user density, interference management, and system throughput in multi-user scenarios.
		\item The intrinsic gain from early stopping fundamentally stems from the disparity in stopping times among users. In fading channel conditions, even when all users transmit the same number of bits at identical power levels, variations in channel fading coefficients result in different perceived SNRs, leading to varying decoding times. We also demonstrate under fading channel scenarios that early stopping concurrently yields both feedback and interference cancellation gains. This shows the robustness of coded water-filling approach.
		\item Finally, we address power allocation in fast fading channel. In a MAC scenario, transmission privileges are granted to the user with the most favorable channel conditions. This requires every user to have knowledge of all other users' channel fading coefficients. However, achieving this in a multi-cell environment is impractical due to the excessive channel state information exchange across cells. To circumvent this, we formulate an optimization model where users can independently determine their power allocation solely based on their own channel fading coefficients. Our findings reveal that, by taking into account interference cancellation among multiple users, elevating the power allocation threshold contributes to interference reduction, therefore enhancing the overall system capacity.
	\end{enumerate}
	
	\subsection{Organizations}
	
The remainder of this paper is organized as follows. Section \ref{back} revisits the concepts of variable-length feedback codes and water-filling solutions, providing a comprehensive overview for both single-user environments and MAC. Section \ref{sec:Ssmu} introduces a multi-cell VLSF model, illustrating the dual benefits of feedback gain and interference cancellation. We further extend this model to a queuing context. In Section \ref{sec:fading}, we provide robust evidence that in fading channel conditions, the gains derived from feedback and interference cancellation remain steadfast, even when all users share symmetric characteristics. Section \ref{sec:waterfilling} introduces an innovative water-filling methodology specifically designed for multi-user scenarios, where power allocation decisions are based exclusively on each user's individual channel fading coefficients. Finally, Section \ref{sec:conclusion} summarizes the key findings of our study.
	
\section{Background}
\label{back}
\subsection{Variable-Length Feedback Codes}

In this paper, we consider memoryless channels characterized by input and output alphabets $\MA$ and $\MB$, respectively. $X = (X_{1},\dots,X_{n},\dots)$ and $\bar{X} = (\bar{X}_1,\dots,\bar{X}_n,\dots)$ denote the true and false codewords, both of which are independent and identically distributed. $Y = (Y_1,\dots,Y_n,\dots)$ represents the channel output. Given that the codeword length is variable, the transmission and reception sequences may extend to infinite lengths, accommodating the flexibility inherent in variable-length coding schemes.

\begin{definition} \cite{poly2,gassian1}
An $(\ell, M, P, \varepsilon)$ VLSF code, where $\ell, P$ are positive real represent average length and power constrain respectively, $M$ is a positive integer of the number of message and error probability $\varepsilon$, is defined by:

1. A finite space $\mathcal{U}$ and a probability distribution $P_U$ on it, defining a random variable $U$ which is revealed to both transmitter and receiver before the start of transmission; i.e. $U$ acts as common randomness used to initialize the encoder and the decoder before the start of transmission.

2. A sequence of encoders $f_{n} : \MU \times \{1, \dots ,M\} \times \MA^{n-1} \to \MB, n \geq 1$, defining channel inputs
\begin{equation}
X_{n} = f_{n}(U, W),
\end{equation}
where $W \in \{1, \dots ,M\}$ is the equiprobable message.

3. A sequence of decoders $g_{n} : \MU \times \MB^n \to \{1, \dots ,M\}$ providing the best estimate of $W$ at time $n$.

4. A non-negative integer-valued random variables $\tau$, a stopping time of the filtration $\mathcal{G}_n = \sigma\{U, Y_{1}, \dots ,  Y_{n}\}$, which satisfies
\begin{equation}
\EE[\tau] \leq \ell.
\end{equation}

5. The expected power constraints at the encoders
\begin{equation}
\sum_{n = 1}^{\infty} \EE[X_n^2] \leq \EE[\tau] P.
\end{equation}

The final decision $\hat{W}$ is computed at the time instant $\tau$:
\begin{equation}
\hat{W} = g_{\tau}(U, Y^{\tau})
\end{equation}
and must satisfy
\begin{equation}
\PP[\hat{W}\neq W] \leq \varepsilon.
\end{equation}
\end{definition}
The finite length achievability coding rate for a VLSF code can be approximated by:
\begin{equation}
\log M_f^*(\ell,\epsilon) = \frac{\ell C}{(1-\epsilon)} + O( \log \ell)
\end{equation}
	
What's more, a general achievability bound of single user model is proposed by \cite{poly2}
\begin{equation}
\varepsilon \le (M-1) Pr\left[\bar{\tau}\le \tau \right]
\end{equation}

where $\bar{\tau}=\inf\{n\ge 0:i(\bar{X}^n;Y^n)\ge \gamma\}$, $\tau=\inf\{n\ge 0:i(X^n;Y^n)\ge \gamma\}$ and $i(x^n;y^n)=log \frac{dP_{Y^n|X^n}(y^n|x^n)}{dP_{Y^n}(y^n)}$.

\subsection{Water-Filling Solution}
In \cite{water2}, the Lagrange multiplier method was used to solve the optimization problem to maximize the mutual information for $S$ parallel Gaussian channels: $Y_s=X_s+Z_s$, $s=1,\dots,S$ with $Z_s \sim \mathcal{N}(0,N_s)$ and $\mathop{\sum}_{s=1}^S X_s^2 \le P$, then the maximum capacity and the optimal power allocation are
	\begin{equation}
		C=\mathop{\sum}_{s=1} ^S\frac{1}{2}\log \left(1+\frac{(v-N_s)^+}{N_s}\right).
	\end{equation}
	where $v$ satisfies $\sum(v-N_s)^+=P$.

	\cite{constant1} considers a discrete-time fading channel model for single user:
	\begin{equation}
		Y_t = h_t \cdot X_t +Z_t,
	\end{equation}
	where $t$ is the discrete-time index, $X_t$ and $Y_t$ are the input and output signals respectively, $Z_t$ is the additive white Gaussian noise with distribution $\mathcal{CN}(0, N)$, indicating that its real and imaginary components are independent and identically distributed (i.i.d.) Gaussian random variables with zero mean and variance $N/2$ each. The time-varying channel fading coefficient is represented by $h_t$. Under the assumption that $h_t$ is perfectly known to both the transmitter and the receiver, and with the noise power normalized to $N=1$, the received signal-to-noise ratio (SNR) at any given time $t$ is calculated as $\gamma_t = |h_t|^2$.

 By viewing the SNR fluctuations induced by random channel fading coefficients as a parallel Gaussian channel, the maximum capacity is calculated as
	\begin{equation}\label{capacity_fading}
		C(P) =  \max \limits_{\int_\gamma P(\gamma)p(\gamma)d \gamma =P} \int_{\gamma} \log (1+P(\gamma)\gamma)p(\gamma) d\gamma.
	\end{equation}
Where $P$ signifies the average power constraint, $P(\gamma)$ represents the power allocated at a given received fading gain $\gamma$, and $p(\gamma)$ denotes the probability distribution of the received fading gain. The authors substantiated both the achievability and converse components of the coding theorem by leveraging the finite division of channel fading statistics. The power allocation strategy aimed at maximizing the expression (\ref{capacity_fading}) is shown as follows:

	\begin{equation}\label{PA_fading}
		P(\gamma)=\left\{ \begin{aligned}
			& \frac{1}{\lambda_0}-\frac{1}{\gamma},& \gamma\ge \lambda_0\\
			& 0, & \gamma<\lambda_0\\
		\end{aligned},\right.
	\end{equation}
where $\lambda_0$ is a certain threshold. Substituting (\ref{PA_fading}) into  (\ref{capacity_fading}), we obtain $\lambda_0$ satisfies
	\begin{equation}
\int_{\lambda_0}^\infty \left(\frac{1}{\lambda_0}-\frac{1}{\gamma}\right)p(\gamma)d\gamma = P,
	\end{equation}
and the channel capacity is
	\begin{equation}\label{capacity_fading1}
		C(P) =  \int_{\lambda_0}^\infty \log (\frac{\gamma}{\lambda_0})p(\gamma) d\gamma.
	\end{equation}

	Since the calculation of the optimal water-filling method is complicated, an approximate water-filling method with constant power allocation is proposed in \cite{constant2}. Using the lower bound of the dual optimization problem, the authors prove that the following constant power allocation is very close to the optimal solution:
	
	\begin{equation}
	P(\gamma)=\left\{ \begin{aligned}
		& P_0,& \gamma\ge \lambda_0\\
		& 0, & \gamma<\lambda_0\\
	\end{aligned},\right.
    \end{equation}
if not too few subchannels are used: $\min\limits_\gamma \{P(\gamma)+\frac{1}{\gamma}\} = P_0 +\min\limits_{\gamma} \{\frac{1}{\gamma}\}$. They also give a low complexity algorithm for realizing the smallest duality gap to optimal solution. And simulation shows that the approximate water-filling method is very close to the optimal water-filling method.
	
In \cite{constant3}, water-filling was extend to MAC scenario
	\begin{equation}
		Y_t = \sum_{k=1}^K h_{t,k} \cdot X_{t,k} +Z_t.
	\end{equation}
 The power constraint is that each user has an average power of $P$ as

\begin{align}
\notag
\int_{\gamma_1}\cdots \int_{\gamma_K} P_k(\gamma_1,\dots,\gamma_K)p(\gamma_1,\dots,\gamma_K)& d\gamma_1 \cdots d\gamma_K =P, \\
& 1 \leq k \leq K,
\end{align}
where $\gamma_k$ is the received fading gain of user-$k$, $P_k(\gamma_1,\dots,\gamma_K)$ is the allocated power of user-$k$.

The objective is to maximize the sum capacities
\begin{align}
\notag
C_K(P) = \int_{\gamma_1}\cdots \int_{\gamma_K} & \log \left(1+\sum_{k=1}^K \gamma_k P_k(\gamma_1,\dots,\gamma_K)\right) \cdot \\
& p(\gamma_1,\dots,\gamma_K)d\gamma_1 \cdots d\gamma_K.
\end{align}
The power allocation strategy of user-$k$ for maximizing the sum capacities is as follows:
	\begin{equation}
		P_k(\gamma_1,\dots,\gamma_K)=\left\{ \begin{aligned}
			& \frac{1}{\lambda_k}-\frac{1}{\gamma_k},& \gamma_k\ge \lambda_k,\gamma_k\ge \frac{\lambda_k}{\lambda_j}\gamma_j,j\ne k\\
			& 0, & \text{otherwise}\\
		\end{aligned},\right.
	\end{equation}
where $\lambda_k$ is the threshold of user-$k$. Assuming symmetry among the users, we have $\lambda_i=\lambda_j$. The condition $\gamma_k > \gamma_j, \forall j \ne k$ signifies that, at any given instant, the exclusive entity permitted to transmit is the user boasting the highest instantaneous power. Meanwhile, all other users are mandated to maintain silence, awaiting their turn until one among them ascends to the position of possessing the greatest power.

\section{Multi-User Interference Cancellation in AWGN Channel}\label{sec:Ssmu}

In this part, we consider interference cancellation for multiple users distributed across multiple cells. For simplicity, we assume that there is only one active user per-cell at any given time, and all users transmit their messages synchronously at the same power level, treating transmissions from users in other cells as interference.

Now, let us introduce the mathematical framework. Assume there are $S$ users distributed at $S$ cells. The $s$-$th$ user is going to transmit a message taking values in $\{1,\dots, M_s=2^{K_s}\}$ uniformly. Without loss of generality, we always assume that $K_s = \alpha_s K$ where $0<\alpha_1<\dots<\alpha_S$, and $K$ tends to infinity. This representation encapsulates the fact that the volume of information (measured in bits) transmitted by each user is different. Therefore, it is instructive to establish a decoding order that aligns with the varying quantities of information transmitted by the users. Without loss of generality, we can assume that decoding starts from the $1$-$st$ user, who has the least number of bits to transmit. Successively, the remaining users are decoded in an ascending order of information bit sequence length, leading up to the completion of decoding for the $S$-$th$ user, who transmits the largest amount of information.

The $s$-$th$ user employs an $(\ell_s, M_s, P, \varepsilon)$ VLSF code for transmission. The channel model for the $s$-$th$ cell can be described as follows:
\begin{equation}
\label{eq:channel}
Y^s_n = X^s_n + \sum_{i\neq s} X^i_n + Z^s_n,
\end{equation}
where $X^s_n$ signifies the $n$-th transmitted symbol from the $s$-th user, subject to the power constraint $E(X^s_n)^2 \leq P$, ensuring that the average power of the transmitted symbols does not exceed the specified limit $P$. When the $i$-th user decides to stop transmission at time $\tau_i$, it is implied that subsequent symbols $X^i_{\tau_i+1},X^i_{\tau_i+2}\dots$ are all zeros, reflecting the termination of transmission. The noise component $Z^s_n \sim \mathcal{N}(0,1)$ denotes \emph{i.i.d.} additive Gaussian noise, characterized by a zero mean and unit variance. $Y^s_n$ is the $n$-th symbol detected by the receiver-$s$. When there are $s$ users concurrently transmitting messages, the SINR is calculated as $P_s = \frac{P}{1+(s-1)P}$, taking into account the interference from the other $s-1$ active users. And the corresponding channel capacity is $C_s = \frac{1}{2} \log(1+P_s)$.

For fixed-length codes, all users stop simultaneously, hence the interference intensity is always at its maximum; for variable-length codes, users stop immediately upon successful transmissions, thereby automatically mitigating interference to other users whose transmissions have not yet completed. The following Fig.\ref{fig1.11} illustrates the distinction between fixed-length codes and variable-length codes. It can be observed that, under multi-user scenarios, the prompt early termination of variable-length codes enables interference cancellation, thereby enhancing the overall system capacity.

\begin{figure*}[htbp]
\centerline{\includegraphics[width = 0.9\textwidth,trim=120 150 120 150,clip]{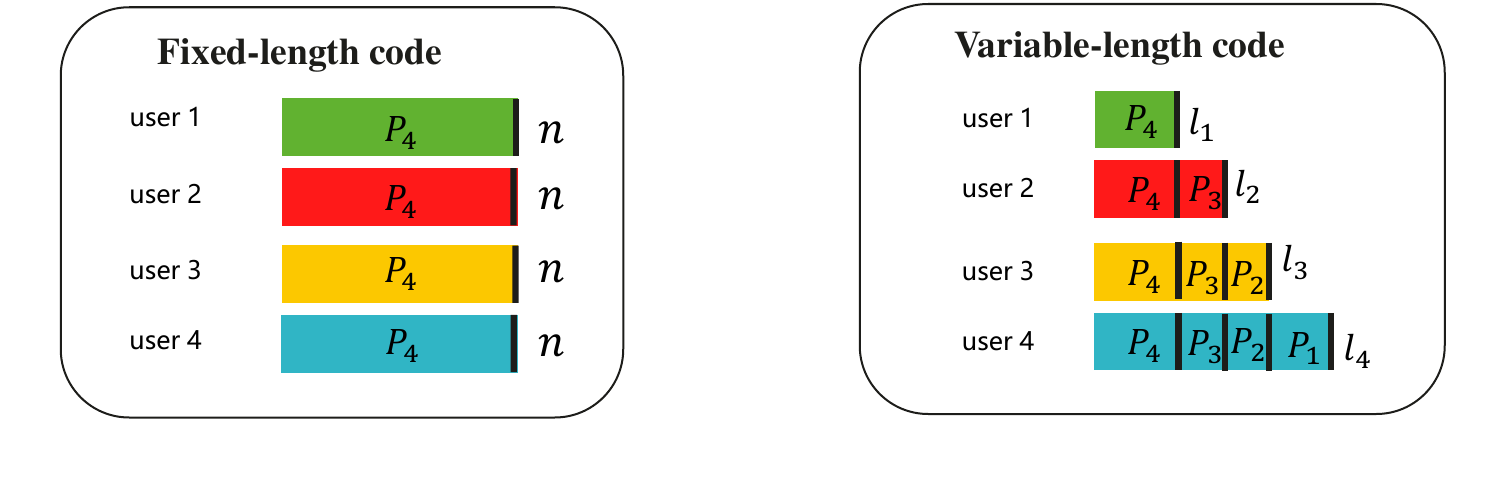}}
\caption{Comparison of fixed-length codes and variable-length codes in multi-user interference cancellation scenarios.}
\label{fig1.11}
\end{figure*}

Initially, all $S$ users transmit their message concurrently. the $s$-$th$ receiver ascertains the termination time $\tau_s$ for $s$-$th$ user's transmission based on the received information, prompting users to stop their message delivery upon successful decoding. In scenarios where $s$ users are actively transmitting, the information density function, denoted as $i_s(X;Y)$, represents the relationship between the transmitted and received signals. Subsequently, we proceed to analyze the average transmission length $\ell_s$ for the $s$-th user.

\begin{theorem}\label{Thm1.1}
There exists a series of  $(\ell_s, M_s=2^{K_s}, P, \varepsilon)$  VLSF codes satisfying
\begin{equation}
\ell_s = (1  -\varepsilon) \sum_{t=1}^{s} \frac{ K_t -  K_{t-1}}{C_{S-t+1}} + O(\log K).
\end{equation}
\end{theorem}

\begin{proof}
\textbf{Encoding:}
For the $s$-th user, denote $\bm{C}^{w,s} = (C^{w,s}_1, C^{w,s}_2,\dots)\in\RR^{\infty}$ for $w\in \{1,\dots, M_s\}$ as the codeword designated for transmitting the $w$-th message. This infinitely long codeword sequence encapsulates the rateless property of variable-length codes, reflecting their adaptability to the specific requirements of the transmission. We employ a random coding technique, whereby the codeword elements $C^{w,j}_n$ are modeled as \emph{i.i.d.} random variables drawn from a Gaussian distribution with zero mean and variance $P$. Let $m$ represent the message to be transmitted. The encoder maps $m$ to the sequence $f_{n,s}(m) \triangleq X_n^s = C^{m,s}_n$, where $C^{m,s}_n$ corresponds to the $n$-$th$ symbol of the codeword. The symbols $X^s_1 ,\dots,X^s_n,\dots$, are conveyed to the receiver through the channel model (\ref{eq:channel}).

\textbf{Decoding:}
Upon receipt of the corresponding symbols $\bm{Y}^{s,n} = (Y^s_1 ,\dots,Y^s_n)$, the decoder computes the accumulated information density
\begin{equation}
S^{w,s}_{n} \triangleq i((C^{w,s}_1,\dots,C^{w,s}_n); \bm{Y}^{s,n}), 1 \leq w \leq M_s.
\end{equation}
Let $\tau^{w,s} \triangleq \inf\{n\geq 0: S^{w,s}_{n} > \gamma^s\}$
be the moment when the accumulated information density for the $w$-$th$ message first exceeds the threshold $\gamma^s$. Define the stopping time $\tau^s = \min_{w} \tau^{w,s}$. The final decision is $\hat{m} = \max\{w: \tau^{w,s} =  \tau^{s}\}$. The aforementioned decoding process signifies that the decoder will output the message whose accumulated information density first reaches the threshold  $\gamma^s$ as the decoding output. This mechanism is fundamental in scenarios employing variable-length coding schemes, where messages are transmitted until the decoder gathers sufficient information to confidently reconstruct a message. By outputting the message that first meets the threshold criterion, the system ensures that the decoding process is both efficient and reliable, striking a balance between minimizing transmission time and maintaining data integrity.

Partition $[0,\tau^s]$ into distinct intervals $\tau^0=0<\tau^1<\dots<\tau^s$, during the time interval $[\tau^{S-t}, \tau^{S-t+1}]$, there are exactly $t$ active users. Therefore,
\begin{align}
\notag
S^{w,s}_{n}  & = \sum_{i=\tau^0+1}^{\tau^1} i_{S}(C_i^{w,s};Y^s_i) + \sum_{i=\tau^1+1}^{\tau^2} i_{S-1}(C_i^{w,s};Y^s_i) + \dots \\
& +  \sum_{i=\tau^{s-1}+1}^{n} i_{S-s+1}(C_i^{w,s};Y^s_i).
\end{align}

\textbf{Error probability:}

Without loss of generality, assume the $1$-$st$ message is transmitted, due to the symmetry introduced by random coding, the codewords of wrong messages ${2,\dots,M_s}$ exhibit identical distribution. From \cite[(111)-(118)]{poly1}, the average error probability is
\begin{equation}\label{eq1.3}
\PP[\bigcup\limits_{w \geq 2} \{\tau^{w,s} = \tau^s\}] \leq (M_s-1)\PP[\tau^{2,s} < \infty] < M_s e^{-\gamma^s} = \frac{1}{ K_s},
\end{equation}
where we choose $\gamma^s=K_s-\log K_s$.

\textbf{Average length:}
For convenience, denote $ \gamma_0 = 0, M_0 = 1$. Since $\tau^1, \dots, \tau^{t-1}$ are independent with $i_{S-t+1}(C_i^{1,s};Y^s_i), i >\tau^{t-1}$,  from Wald's identity and \cite[Lemma 1]{gassian1},
\begin{align}
\notag
\gamma^s+ O(1) & = \EE[S^{1,s}_{\tau^{1,s}}] = \EE[\tau^1] C_S + \EE[\tau^2 - \tau^1] C_{S-1} + \dots \\  
& +\EE[\tau^{1,s} - \tau^{s-1}]  C_{S-s+1}.
\end{align}

Therefore,
\begin{align}\label{eq1.2}
\notag
\EE[\tau^{s}] & \leq \EE[\tau^{1,s}] \\
& = \frac{1}{C_{S-s+1}} (\gamma^s + \sum_{t=1}^{s-1} \EE[\tau^t] (C_{S-t} -   C_{S-t+1})) + O(1).
\end{align}

Note that $\EE[\tau^1] \leq \EE[\tau^{1,1}]=\frac{\gamma^1}{C_S} + O(1)$.  Assume  $\EE[\tau^k] = \sum_{t=1}^k \frac{\gamma^t - \gamma^{t-1}}{C_{S-t+1}} + O(1)$ for $1\leq k\leq s-1$. Then
\begin{align}
\notag
\EE[\tau^s] & \leq \EE[\tau^{1,s}] \\ 
& = \frac{1}{C_{S-s+1}}(\gamma^s + \sum_{k=1}^{s-1} \EE[\tau^k] (C_{S-k} -   C_{S-k+1}) + O(1) ) \\
\notag
& = \frac{1}{C_{S-s+1}}(\gamma^s + \sum_{k=1}^{s-1} \sum_{t=1}^{k} \frac{\gamma^t - \gamma^{t-1}}{C_{S-t+1}} (C_{S-k} -   C_{S-k+1}) ) \\
& + O(1)  \\
\notag
& = \frac{1}{C_{S-s+1}}(\gamma^s + \sum_{t=1}^{s-1} \frac{\gamma^t - \gamma^{t-1}}{C_{S-t+1}} \sum_{k=t}^{s-1} (C_{S-k} -   C_{S-k+1}) ) \\
& + O(1)  \\
& = \frac{1}{C_{S-s+1}}(\gamma^s + \sum_{t=1}^{s-1} \frac{\gamma^t - \gamma^{t-1}}{C_{S-t+1}} (C_{S-s+1} -   C_{S-t+1}) ) \\
\notag
& + O(1)  \\
& = \frac{\gamma^s}{C_{S-s+1}} + \sum_{t=1}^{s-1} \frac{\gamma^t - \gamma^{t-1}}{C_{S-t+1}} - \sum_{t=1}^{s-1} \frac{\gamma^t - \gamma^{t-1}}{C_{S-s+1}} + O(1)  \\
& = \sum_{t=1}^{s} \frac{\gamma^t - \gamma^{t-1}}{C_{S-t+1}} + O(1).
\end{align}

Recall that $\gamma^s =  K_s - \log  K_s$. Thus,
\begin{equation}\label{eq1.4}
\EE[\tau^s] = \sum_{t=1}^{s} \frac{ K_t -  K_{t-1}}{C_{S-t+1}} + O(\log K).
\end{equation}

Hence, there exists a series of $(\ell'_s, M_s, \frac{1}{ K_s})$ VLSF code satisfying
\begin{equation}
\ell'_s = \sum_{t=1}^{s} \frac{ K_t -  K_{t-1}}{C_{S-t+1}} + O(\log K).
\end{equation}

Now, let's consider a scenario where the user transmits a codeword from an $(\ell'_s, M_s, \frac{1}{ K_s})$ code with a probability of $\frac{ (1-\varepsilon) K_s}{ K_s - 1}$, and refrains from transmitting anything in the complementary case. This particular scheme ensures an error probability that is bounded above by $\varepsilon$ and yields an average codeword length of $\ell_s = \frac{ (1-\varepsilon) K_s}{ K_s - 1} \ell_s'$. We have
\begin{equation}
\ell_s = (1  -\varepsilon) \sum_{t=1}^{s} \frac{ K_t -  K_{t-1}}{C_{S-t+1}} + O(\log K).
\end{equation}
\end{proof}

\begin{remark}
\begin{align}
\ell_s & =  (1  -\varepsilon) \sum_{t=1}^{s} \frac{ K_t -  K_{t-1}}{C_{S-t+1}} + O(\log K) \\
& <  (1  -\varepsilon) \sum_{t=1}^{s}  \frac{ K_t -  K_{t-1}}{C_{S}} + O(\log K)  \\
& =  (1  -\varepsilon) \frac{K_s}{C_S} + O(\log K).
\end{align}

Hence, $\log M_s > \frac{C_S}{1  -\varepsilon}\ell_s + O(\log \ell_s)$, the inequality represents the relationship between the average length of variable-length codes and the number of message bits, without taking into account the benefits of interference cancellation. Consequently, in multi-user scenarios, early termination enables the simultaneous realization of gains from both feedback and interference cancellation.
\end{remark}

In Fig.\ref{fig1.1}, we present a comparison of the codeword lengths between fixed-length codes and variable-length codes. The length of the variable-length code is derived from Theorem \ref{Thm1.1}. For fixed-length codes, the codeword lengths for all users are identical; thus, the length must be sufficiently large to ensure that all users can decode their messages with a small error probability. This necessitates that the codeword length $n$ for fixed-length codes satisfies the condition $\log M_S \leq nC_S - \sqrt{nV_S}Q^{-1}(\epsilon)$, where $C_S$ and $V_S$ are the channel capacity and channel dispersion with the SINR $P_S$, respectively. The figure vividly demonstrates that the average length of the variable-length code is significantly shorter than that of the fixed-length code. Moreover, the disparity in average codeword length savings becomes more pronounced as the difference in the number of messages to be transmitted among users increases. This phenomenon attributes to the ability of variable-length codes to enable rapid early termination for users with shorter messages, which in turn facilitates interference cancellation. This reduction in interference leads to higher capacity for subsequent users, allowing for more efficient communication and higher system throughput.

\begin{figure}[h]
	\begin{center}
		\subfigure[\scriptsize $K_1$=300, $K_2$=1000.]{\includegraphics[width=0.4\textwidth, clip, trim= 100 265 120 265,clip]{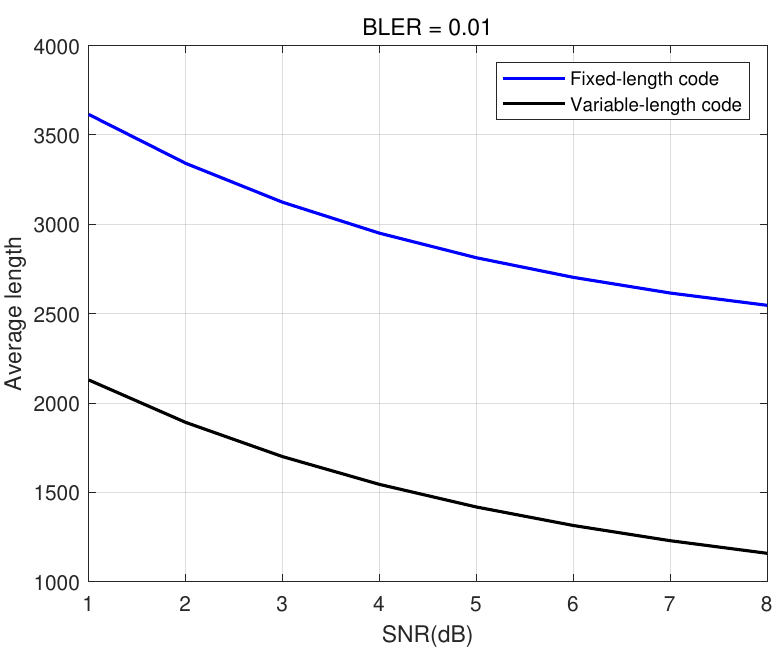}}
		\subfigure[\scriptsize $K_1$=600, $K_2$=1000.]{\includegraphics[width=0.4\textwidth, clip, trim= 100 265 120 265,clip]{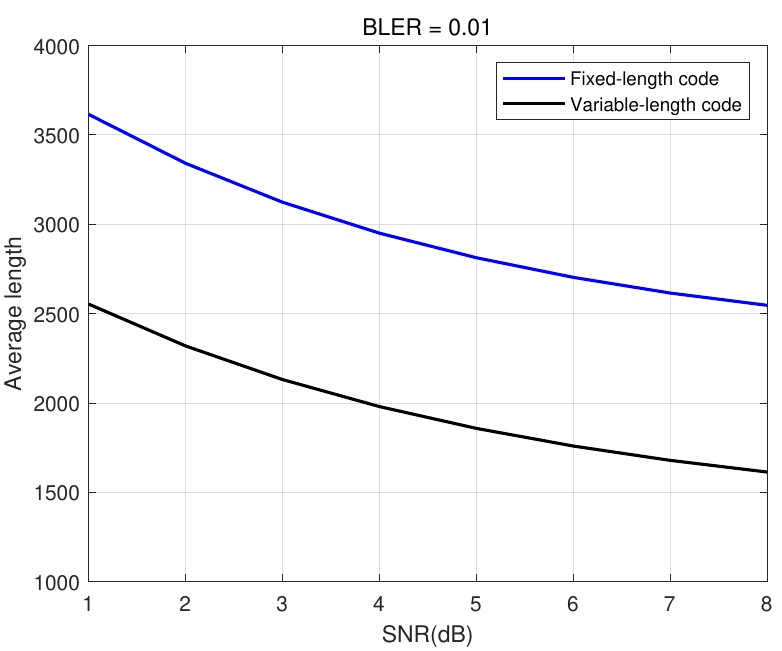}}
	\end{center}
	\caption{Comparison of code lengths between fixed-length codes and variable-length codes.}\label{fig1.1}
\end{figure}

Now, we extend Theorem \ref{Thm1.1} to a queuing model. A new transmission request arrives at fixed time intervals $T_{sub}$, and within the $s$-th cell, $(\ell_s, M_s, P, \varepsilon)$ VLSF codes are employed to convey messages. If the decoding of the preceding message fails to conclude within the $T_{sub}$ time frame, the subsequent message must wait, leading to a situation akin to queuing or congestion. This queuing model introduces a practical constraint on the transmission process, reflecting real-world scenarios where there is a continuous stream of data packets requiring transmission. A simplified illustration of the queuing model is shown in Fig.\ref{fig 21}. Subsequently, we proceed to analyze the average code length within the queueing model.

\begin{figure}[htbp]
\centerline{\includegraphics[width = 0.5\textwidth,trim=310 90 200 100,clip]{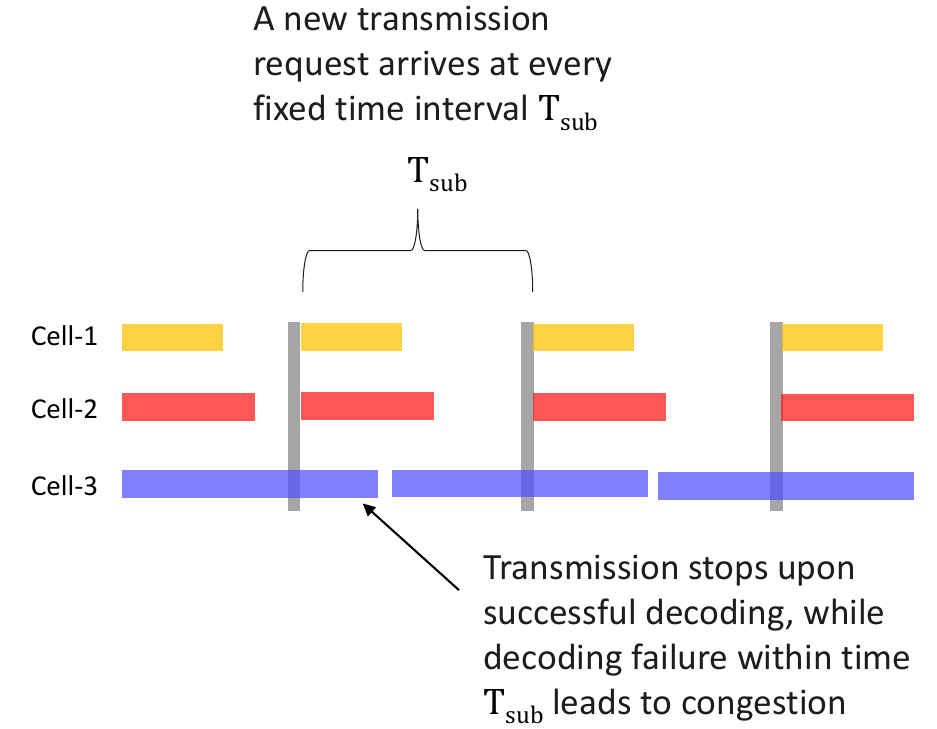}}
\caption{A illustration of the queuing model.}
\label{fig 21}
\end{figure}

\begin{theorem}\label{Thm1.2}
Define $I \triangleq  T_{sub}C_{S-r} - \sum_{t=1}^{r} \frac{C_{S-r+1}}{C_{S-t+1}}(K_t - K_{t-1}) + K_{r}$ as the average amount of information that can be accumulated during each $T_{sub}$ time interval, and $r$ as the largest index for which the average codeword length $\leq T_{sub}$. For $s > r$, the decoding in the $s$-$th$ cell need multiple intervals, let  $c_s \triangleq  \lfloor \frac{\gamma^s}{I} \rfloor$, $b_s \triangleq \gamma^s - c_s I$ be the number of complete time intervals and the remaining average amount of information required after $c_s$ time intervals, respectively.

 There exists a series of  $(\ell_s, M_s, P, \varepsilon)$ VLSF code satisfying
\begin{equation}\label{eq2.1}
\ell_s = (1  -\varepsilon) \sum_{t=1}^{s} \frac{K_t - K_{t-1}}{C_{S-s+1}} + O(\log K).
\end{equation}
for $j\leq r$, and
\begin{equation}\label{eq2.2}
\ell_s \leq (1  -\varepsilon) (c_s T_{sub} + \frac{b_s-K_{r_s}}{C_{S-r_s}} + \sum_{t=1}^{r_s} \frac{K_t - K_{t-1}}{C_{S-t+1}}) + O(\log K).
\end{equation}
for $j> r$. Where $r_s$ satisfies that the decoding will stop between $\tau^{r_s}$ and $\tau^{r_s+1}$ at the $c_s+1$-$th$ interval.
\end{theorem}

\begin{proof}

The encoding, decoding and the error probability  analysis parts are the same as that in Theorem \ref{Thm1.1}. Let $0\leq r\leq S$ to be the largest integer such that
\begin{equation}
\ell'_r = \sum_{t=1}^{r} \frac{ K_t -  K_{t-1}}{C_{S-t+1}}<T_{sub}.
\end{equation}
 Then the first $r$ cells will not encounter congestion, according to Theorem \ref{Thm1.1}, there exists  a series of $(\ell_s, M_s, \varepsilon)$ VLSF code satisfying (\ref{eq2.1}).

For the $s$-th user where $s>r$, since we focus on establishing an upper bound for $\ell_s$, consider the scenario that represents the most challenging channel conditions. Specifically, we assume that the users initiate message transmission at the very start of a particular time period. It is evident that the average  length will be shorter if transmissions can stop at any other point within the time period.

Firstly, we calculate the average amount of information that can be accumulated during each $T_{sub}$ time interval assuming that all the users initiate message transmission at the very start of a particular time period. Recall that $r$ is the largest index for which the average codeword length $\ell'_r \leq T_{sub}$. In one time period, the accumulated information density is
\begin{align}
\notag
S^{w,s}_{T_{sub}} & = \sum_{i=\tau^0+1}^{\tau^1} i_{S}(C_i^{w,s};Y^s_i) + \sum_{i=\tau^1+1}^{\tau^2} i_{S-1}(C_i^{w,s};Y^s_i) + \dots \\
& +  \sum_{i=\tau^{r}+1}^{T_{sub}} i_{S-r}(C_i^{w,s};Y^s_i).
\end{align}

The average accumulated information is
\begin{align}
I & = \EE[S^{w,s}_{T_{sub}}] \\
& = \sum_{k=1}^{r} \EE[\tau^k - \tau^{k-1}] C_{S-k+1} + (T_{sub} -\EE[ \tau^r]) C_{S-r} \\
& = T_{sub}C_{S-r} + \sum_{k=1}^{r} \EE[\tau^k] (C_{S-k+1} -   C_{S-k})  \\
& = T_{sub}C_{S-r} + \sum_{k=1}^{r} \sum_{t=1}^{k} (\frac{\gamma^t - \gamma^{t-1}}{C_{S-t+1}}+ O(1))(C_{S-k+1} -  C_{S-k})  \\
& = T_{sub}C_{S-r} + \sum_{t=1}^{r} \frac{\gamma^t - \gamma^{t-1}}{C_{S-t+1}} \sum_{k=t}^{r} (C_{S-k+1} -  C_{S-k})  +  O(1) \\
& = T_{sub}C_{S-r} + \sum_{t=1}^{r} \frac{\gamma^t - \gamma^{t-1}}{C_{S-t+1}} (C_{S-t+1} -  C_{S-r})  +  O(1) \\
& =  T_{sub}C_{S-r} - \sum_{t=1}^{r} \frac{C_{S-r+1}}{C_{S-t+1}}(\gamma^t - \gamma^{t-1}) + \gamma_{r} +  O(1) \\
& = T_{sub}C_{S-r} - \sum_{t=1}^{r} \frac{C_{S-r+1}}{C_{S-t+1}}(K_t - K_{t-1}) + K_{r} +  O(\log K)
\end{align}

Let $c_s =  \lfloor \frac{\gamma^s}{I} \rfloor$, $b_s = \gamma^s - c_s I$, the user accumulates $c_sI$ information over the first $c_s$ time intervals, at the $(c_s+1)$-$th$ period, there is at most $b_s$ information to collect. At time $c_s I+\tau^j$, the user collects an average of $\gamma^j$ information at the $(c_s+1)$-$th$ period, hence let $r_s$ be the greatest index such that $\gamma^{r_s} < b_s$, the $s$-$th$ user stops in time interval $[\tau^{r_{s}}, \tau^{r_{s+1}}]$.

Similar as Theorem \ref{Thm1.1},
\begin{align}
\EE[\tau^*_j] & \leq \EE[\tau_j]  \\
\notag
& \leq  c_s T_{sub} + \frac{1}{C_{S-r_s}}( b_s  + \sum_{k=1}^{r_s} \EE[\tau^k] (C_{S-k} -   C_{S-k+1} ))\\
& + O(1) \\
\notag
& \leq   c_s T_{sub} + \frac{b_s}{C_{S-r_s}} + \sum_{t=1}^{r_s} \frac{\gamma_t - \gamma_{t-1}}{C_{S-t+1}} - \sum_{t=1}^{r_s} \frac{\gamma_t - \gamma_{t-1}}{C_{S-r_s}} \\
& + O(1)  \\
& =  c_s T_{sub} + \frac{b_s - \gamma_{r_s}}{C_{S-r_s}} + \sum_{t=1}^{r_s} \frac{\gamma_t - \gamma_{t-1}}{C_{S-t+1}} + O(1) \\
& =  c_s T_{sub} + \frac{b_s-K_{r_s}}{C_{S-r_s}} + \sum_{t=1}^{r_s} \frac{K_t - K_{t-1}}{C_{S-t+1}} + O(\log K).
\end{align}

Therefore, there exists  $(\ell_s, M_s, P, \varepsilon)$ VLSF codes satisfying (\ref{eq2.2}) for $s>r$.

\end{proof}

Fig.\ref{fig23} shows the variation in average codeword length as a function of different time intervals $T_{sub}$. As the interval $T_{sub}$ increases, the figure reveals a reduction in the number of users experiencing congestion, accompanied by a rapid escalation in the average codeword length. To better understand this relationship, we consider the two extremities depicted in Fig.\ref{fig 22}: at one end, when $T_{sub}$ is exceedingly small, congestion becomes a constant phenomenon, and consequently, the advantages of interference cancellation facilitated by early termination are nullified. Conversely, at the other end of the spectrum, when $T_{sub}$ is significantly large, congestion is significantly alleviated, enabling the full exploitation of interference cancellation gains that result from early termination.

\begin{figure}[h]
	\begin{center}
		\subfigure[\scriptsize Small $T_{sub}$.]{\includegraphics[width=0.48\textwidth, trim= 230 200 250 200,clip]{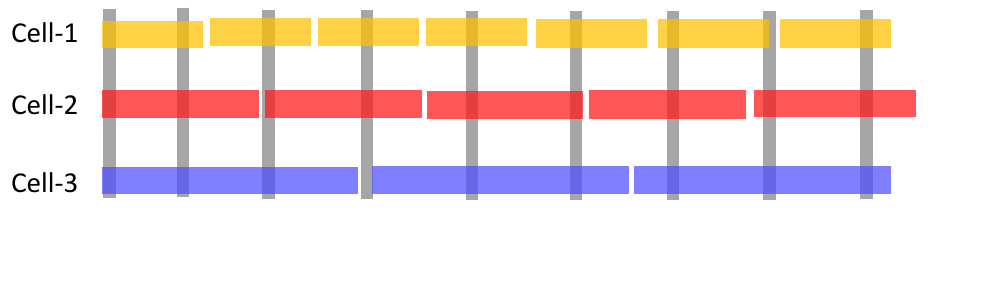}}
		\subfigure[\scriptsize Large $T_{sub}$.]{\includegraphics[width=0.48\textwidth, trim= 230 220 250 150,clip]{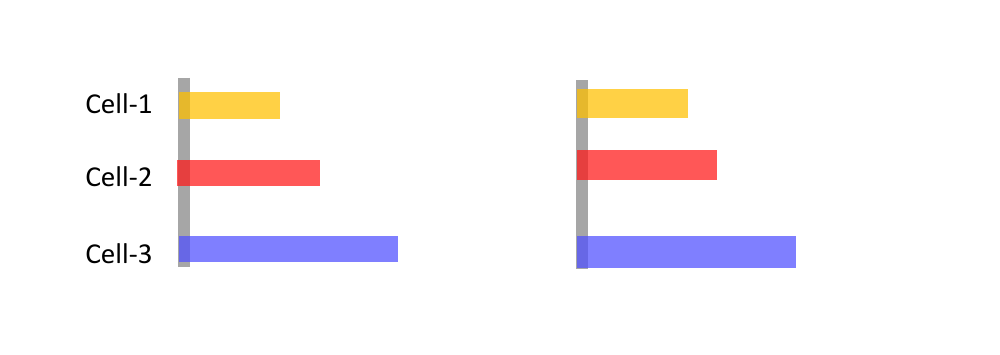}}
	\end{center}
	\caption{Two extreme scenarios: constant congestion and no congestion.}\label{fig 22}
\end{figure}

\begin{figure}[h]
	\begin{center}
		\subfigure[\scriptsize $K_1$=300, $K_2$=1000.]{\includegraphics[width=0.4\textwidth,trim=  100 265 120 265,clip]{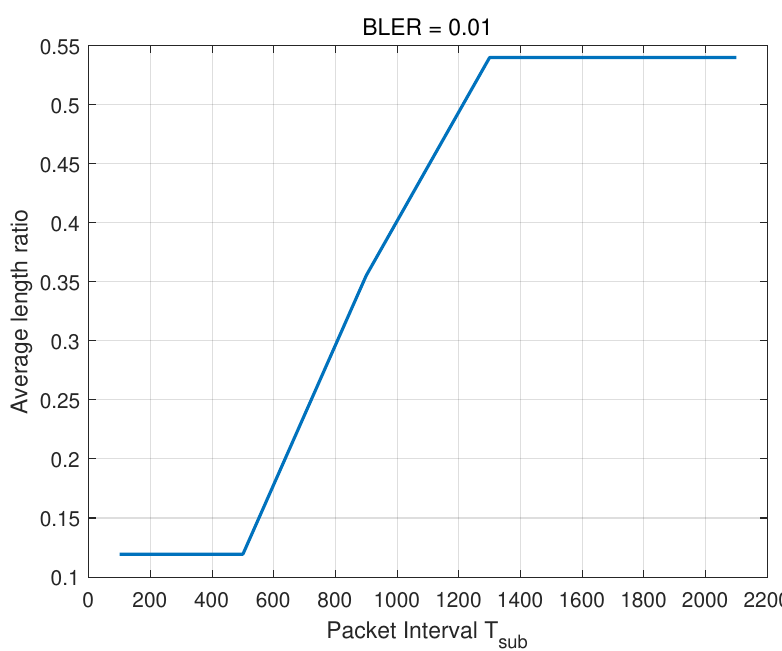}}
		\subfigure[\scriptsize $K_1$=600, $K_2$=1000.]{\includegraphics[width=0.4\textwidth, trim=  100 265 120 265,clip]{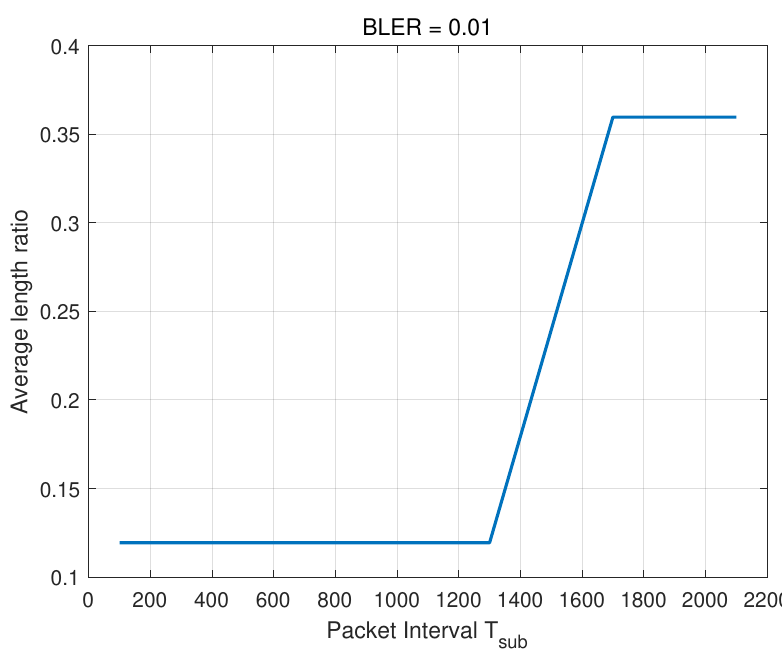}}
	\end{center}
	\caption{The ratio of code length savings under different packet intervals with two cells.}\label{fig23}
\end{figure}

In Fig.\ref{fig23}, the $y$-axis represents the ratio of the codeword length savings achieved by employing variable-length codes in contrast to fixed-length codes, represented as $1 - \frac{l_S}{n_S}$. Here, $l_S$ signifies the average codeword length of variable-length codes required for transmitting $K_S$ bits within cell-$S$ (\ref{eq2.1}) (\ref{eq2.2}), whereas $n_S$ denotes the codeword length of fixed-length codes for conveying the equivalent number of bits within the identical cell. As $T_{sub}$ increases, the initial phase exhibits a stable code length gain, a period marked by persistent congestion that inhibits the realization of interference cancellation benefits. Subsequently, the code length gain experiences an enhancement as the packet interval continues to expand, culminating in a saturation gain. This final gain is achieved when congestion is entirely alleviated, allowing for the full exploitation of interference cancellation gains. Notably, the magnitude of the final saturation gain is directly proportional to the disparity in the number of bits transmitted between cells; the greater this disparity, the higher the ultimate saturation gain realized.

\section{Multi-User Interference Cancellation in Fading Channel}
\label{sec:fading}

In Section \ref{sec:Ssmu}, we consider the model that the $S$ users transmits codes with the same power and different number of information bits. In fact, the interference cancellation gain stems from varying decoding times among users, which can be realized through multiple parameter disparities. For example, different payload lengths lead to varying code rates for decoding, and low-rate codes will be decoded earlier; similarly, different transmission powers also result in diverse decoding times. Even in scenarios where users share the same number of payload bits and transmit at identical power levels, the inherent randomness of fading channels still introduces varying decoding time. The impact of random fading coefficients on the transmission leads to variations in the actual power received. This mechanism ensures that, despite uniform initial conditions, the dynamic environment of fading channels will yield varying decoding time which is necessary for interference cancellation gains.

In this section, the users transmit one of the $M=2^K$ messages uniformly in the block fading interfere channel, i.e., fading coefficients are constant for a single transmission. The channel model is given by
\begin{equation}
Y^s_n = h_{s}X^s_n + \sum_{i\neq s} h_{i}X^i_n + Z^s_n,
\end{equation}
where $h_{i}$ is the fading coefficient for the channel from the user in $i$-$th$ cell, $Z^s_n$ are \emph{i.i.d.} $\mathcal{CN}(0, 1)$ random variables. Similar interference cancellation phenomena as in Fig.\ref{fig1.12} also occur under fading channel even when users have identical number of payload bits and transmission powers.

\begin{figure*}[htbp]
\centerline{\includegraphics[width = 0.9\textwidth,trim=120 150 120 150,clip]{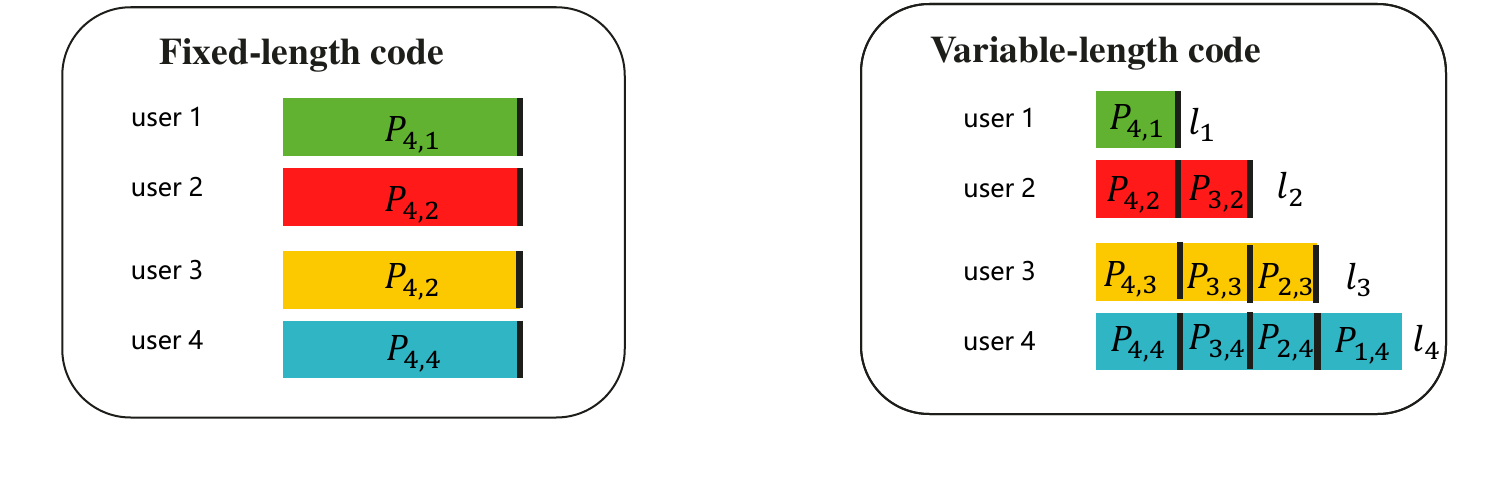}}
\caption{Comparison of fixed-length codes and variable-length codes in multi-user interference cancellation scenarios in fading channel.}
\label{fig1.12}
\end{figure*}

Without loss of generality, assume the decoding order is from user in the $1$-$st$ cell to user in the $S$-$th$ cell, when there are $s$ users transmitting, the SINR for the $j$-$th$ user is denoted as $P_{s,j}=\frac{P|h_{j}|^2}{1+\sum_{t=S-s+1, t\neq j}^{S} P|h_{t}|^2}$, and the corresponding information density functions is denoted as $i_{s,j}(X;Y)$. Denote $C_{s,j} = \EE(i_{s,j}(X;Y)) = \log(1+P_{s,j})$.
\begin{theorem}\label{Thm3.1}
There exists a series of  $(\ell_j, M, P, \varepsilon)$ VLSF codes satisfying
\begin{equation}\label{eq3.1}
\ell_j =  a_j(1-\varepsilon) K + O(\log K)
\end{equation}
where $a_j$ depends on the fading coefficients, their values are given in (\ref{eq3.2}) and (\ref{eq3.3}).
\end{theorem}

\begin{proof}
We use the same encoding and decoding method as Theorem \ref{Thm1.1}. Now since the number of transmitted bits is the same for all users $\gamma^j = \gamma \triangleq K +  \log K$.

Similar to (\ref{eq1.2}), we have
\begin{equation}
\EE[\tau^j] = \frac{1}{C_{S-j+1,j}} (\gamma + \sum_{k=1}^{j-1} \EE[\tau^k] (C_{S-k,j} -   C_{S-k+1,j})) + O(1).
\end{equation}

Denote
\begin{equation}\label{eq3.2}
a_1 = \frac{1}{C_{S,1}}
\end{equation}

and
\begin{equation}\label{eq3.3}
a_{k} = \frac{1}{C_{S-k+1,k}}(1 + \sum_{t=1}^{k-1} a_t (C_{S-t,k} - C_{S-t+1,k}) )
\end{equation}
for $2\leq k\leq S$. Note that $\EE[\tau^1] \leq \frac{\gamma}{C_{S,1}} + O(1) = a_1\gamma + O(1)$.  Assume by induction, $\EE[\tau^k] \leq a_k\gamma + O(1)$ for $1\leq k\leq j-1$. Then
\begin{align}
\notag
\EE[\tau^j] & \leq \frac{1}{C_{S-j+1,j}}(\gamma + \sum_{k=1}^{s-1} (a_k\gamma + O(1)) (C_{S-k,j} -   C_{S-k+1,j}) )\\ 
& + O(1) \\
& \leq \frac{1}{C_{S-j+1,j}}(1 + \sum_{k=1}^{s-1} a_k (C_{S-k,j} -   C_{S-k+1,j}) )\gamma + O(1) \\
& = a_j\gamma + O(1).
\end{align}

Therefore,
\begin{equation}
\EE[\tau^j] = a_j K + O(\log K).
\end{equation}

Hence, there exists a series of $(\ell'_j, M, P, \frac{1}{K})$ VLSF code satisfying
\begin{equation}
l'_j = a_j K + O(\log K).
\end{equation}

Similar as the proof in Theorem \ref{Thm1.1}, there exists  $(\ell_j, M, P, \varepsilon)$ VLSF codes satisfying
\begin{equation}
l_j =  a_j(1-\varepsilon) K + O(\log K).
\end{equation}
\end{proof}

\begin{example}
Consider a scenario where $S$ users are distributed across $S$ cells, with each user transmitting messages within their respective cell at a power level $P$. The fading coefficients, represented as $h_{1},\dots,h_{S}$, are modeled as Rayleigh random variables, implying that the squared magnitudes $|h_{1}|^2,\dots, |h_{S}|^2$ follow exponential distributions with a mean of $1$. Assuming that $|h_{1}|^2>\dots>|h_{S}|^2$, the order statistics of the exponential distribution, as detailed in \cite[Theorem 4.6.1]{order_sta}, dictate that the distribution of $|h_{S-j}|^2$ is equal to as $\sum_{i=0}^j \frac{z_i}{S-i}$, where $z_i, 0 \leq i \leq S-1$, are independent exponential random variables with a mean of $1$. Furthermore, the expected value of $|h_{S-j}|^2$ is given by $\E[|h_{S-j}|^2]=\sum_{i=0}^j \frac{1}{S-i}$.

Recall that when there are $s$ users transmitting, the SINR for the $j$-$th$ user is denoted as $P_{s,j}=\frac{P|h_{j}|^2}{1+\sum_{t=S-s+1, t\neq j}^{S} P|h_{t}|^2}$. We provide a set of typical values for the fading coefficients: $|h_{j}|^2=\E[|h_{j}|^2]$. Therefore, $P_{s,j}=\frac{P\sum\limits_{i=0}^{S-j} \frac{1}{S-i}}{1+P\sum\limits_{t=S-s+1, t\neq j}^{S} \sum\limits_{i=0}^{S-t} \frac{1}{S-i}}$.

In Fig.\ref{fig3.1}, a comparison is illustrated between the average lengths of codewords for fixed-length codes and variable-length codes under a block fading channel model. It is evident that the average length of codewords for the variable-length code is significantly shorter than that of the fixed-length code. This notable gain in efficiency can be attributed to the inherent randomness of the fading channel. The variable-length coding scheme leverages the stochastic nature of the channel to adaptively adjust the codeword length, thereby optimizing resource utilization and minimizing redundancy. This adaptability allows for shorter average codeword lengths, as the scheme can terminate transmissions earlier when channel conditions are favorable, without compromising on the reliability of the communication.
\end{example}	

\begin{figure}[htbp]
\centerline{\includegraphics[width = 0.5\textwidth, clip, trim= 100 265 120 265,clip]{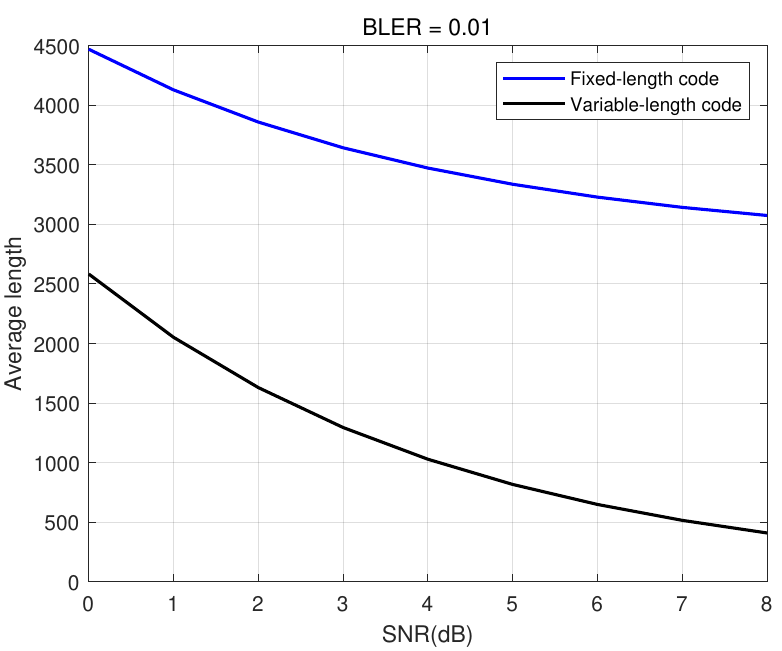}}
\caption{Comparison of average code lengths between fixed-length codes and variable-length codes under fading channel conditions, $K=1000$, $S=2$.}
\label{fig3.1}
\end{figure}

\section{Power Allocation Strategy in Fast Fading Channel}
\label{sec:waterfilling}

In this section, we analyze the case where $S$ users are engaged in transmitting messages within a fast fading channel environment. Each user operates at a common average power level denoted as $P$. The channel model is characterized by
\begin{equation}
Y^s_n = h_{s,n}X^s_n + \sum_{i\neq s} h_{i,n}X^i_n + Z^s_n,
\end{equation}
where $h_{i,n}$ represent \emph{i.i.d.} fading coefficients. Let $\gamma_{s,n}= |h_{s,n}|^2$ denote the received fading gain for the $s$-$th$ user at the $n$-$th$ time slot. The primary objective is to allocate power based on the instantaneous fading gain $\gamma_{s,n}$ in a multi-user scenario, aiming to enhance the overall capacity. Contrary to the classical multiuser water-filling solution presented in \cite{constant3}, our approach simplifies the power allocation strategy in the following two significant ways.

Firstly, we implement a constant power allocation water-filling method. When the fading coefficient $\gamma_{s,n}$ surpasses a certain threshold $\gamma_{th}$, which is subject to optimization, the transmitting power is set to $\frac{P}{\PP(\gamma \geq \gamma_{th})}$, where $\PP(\gamma \geq \gamma_{th})$ represents the probability of the SNR exceeding the threshold. This ensures that power is allocated only when the channel conditions are favorable. Conversely, if $\gamma_{s,n}$ falls below the threshold $\gamma_{th}$, the user refrains from transmitting to minimize interference. It is evident that, under this scheme, the average transmitting power for each user remains precisely at $P$.

Secondly, in a MAC scenario, transmission privileges are typically awarded to the user exhibiting the most favorable channel conditions. This approach, however, requires users to have knowledge of all other users' channel fading coefficients. To circumvent the complexity and impracticality of such information exchange, we develop an optimization model where users in cell-$s$ determine their power allocation based solely on their own fading gain $\gamma_{s,n}$. This user-centric approach reduces the need for global channel state information, making the power allocation strategy more feasible in multi-cell environments.

In summary, our power allocation strategy combines a threshold-based constant power allocation with a user-centric optimization model. It simplifies the implementation while still enhancing capacity in multi-cell scenarios. This approach not only reduces the computational complexity and signaling overhead but also ensures a more equitable distribution of resources among users.

Since each user exhibits symmetry, let us consider the first user as an illustrative example. Given the threshold $\gamma_{th}$, the capacity of user-$1$ can be determined as
\begin{align}\label{eq:4.1}
\notag
C(\gamma_{th}) = & \int_{\gamma_{th}}^{\infty}\cdots \int_{0}^{\infty} \\
\notag
& \log(1+ \frac{\gamma_{1}/ \PP(\gamma_{1}>\gamma_{th})}{1/P + \sum\limits_{i=2}^S \gamma_{i}/ \PP(\gamma_{i}>\gamma_{th}) I(\gamma_{i}>\gamma_{th})})\cdot \\ 
& p(\gamma_{1}) \cdots p(\gamma_{i})d\gamma_{1} \cdots d\gamma_{S}.
\end{align}

If we neglect the interference cancellation benefits from other users and assume they transmit messages at power $P$ with channel gains $\gamma_{i}=1$, The problem is equivalent to a single-user optimization problem when the noise variance is equal to $1/P+S-1$. According to \cite[Section IV]{constant2}, we have that the optimal threshold $\gamma_{single}$ is the solution of
\begin{equation}
\gamma_{th} e^{\gamma_{th}} = 1/P+S-1.
\end{equation}

Now take the interference cancellation of other users into consideration. The capacity lower bound based on one-dimensional integral is
\begin{align}
\notag
C(\gamma_{th}) \geq & \int_{\gamma_{th}}^{\infty} \sum_{t=0}^{S-1} \binom{S-1}{t}\PP(\gamma<\gamma_{th})^{S-1-t} \PP(\gamma\geq \gamma_{th})^{t} \cdot \\
& \log(1+ \frac{\gamma / \PP(\gamma>\gamma_{th})}{1/P + t\bar{\gamma}}) p(\gamma) d\gamma \triangleq {C}_{l}(\gamma_{th}),
\end{align}
where $t$ is the number of active users and $\bar{\gamma}=\int_{r_{th}}^{\infty} \frac{\gamma p(\gamma)}{\PP^2(\gamma\geq \gamma_{th})} d \gamma$ is the average received channel gain condition on $\gamma\geq \gamma_{th}$. The inequality is from that $\log(1+c/x)$ is a convex function and Jensen Inequality.
 Let $\gamma_{multi}$ be the optimized value $\gamma_{th}$ to maximize ${C}_{l}(\gamma_{th})$.

For example, if the fading coefficient $h$ is a Rayleigh distribution, we have
\begin{align}
\notag
C_l(\gamma_{th}) = & \int_{\gamma_{th}}^{\infty} \sum_{t=0}^{S-1} \binom{S-1}{t} (1-e^{-\gamma_{th}})^{S-1-t} e^{-\gamma_{th}t} \cdot \\
& \log(1+ \frac{\gamma}{e^{-\gamma_{th}}/P + t(1+\gamma_{th})}) e^{-\gamma}d\gamma,
\end{align}
where $\bar{\gamma}=(1+\gamma_{th})e^{\gamma_th}$.

Fig.\ref{fig41} showcases the thresholds $\gamma_{single}$ and $\gamma_{multi}$ alongside the capacities (\ref{eq:4.1}) induced by these thresholds and their respective capacity lower bounds $C_{single}=C_l(\gamma_{single})$ and $C_{multi}=C_l(\gamma_{multi})$. A notable observation is that $\gamma_{multi}$ is higher than $\gamma_{single}$, yet it corresponds to a superior capacity. This result can be attributed to the strategic transmission behavior encouraged by the multi-user threshold $\gamma_{multi}$. When users are inclined to transmit messages primarily through favorable channel conditions, the interference imposed on other users is significantly reduced. This reduction in interference stems from the fact that users with bad channel conditions pause their transmissions, leaving the channel less crowded for other users to transmit with less interference. Consequently, the system as a whole experiences enhanced capacity, as users are able to communicate more efficiently and with fewer disruptions. The figure thus highlights the strategic advantage of the multi-user threshold $\gamma_{multi}$ in managing interference and optimizing resource utilization within a multi-user communication environment, leading to improved overall system performance.

\begin{figure}[h]
	\begin{center}
		\subfigure[\scriptsize Capacity and capacity lower bound under threshold $\gamma_{multi}$ and $\gamma_{single}$ respectively.]{\includegraphics[width=0.4\textwidth,trim=  100 265 120 265,clip]{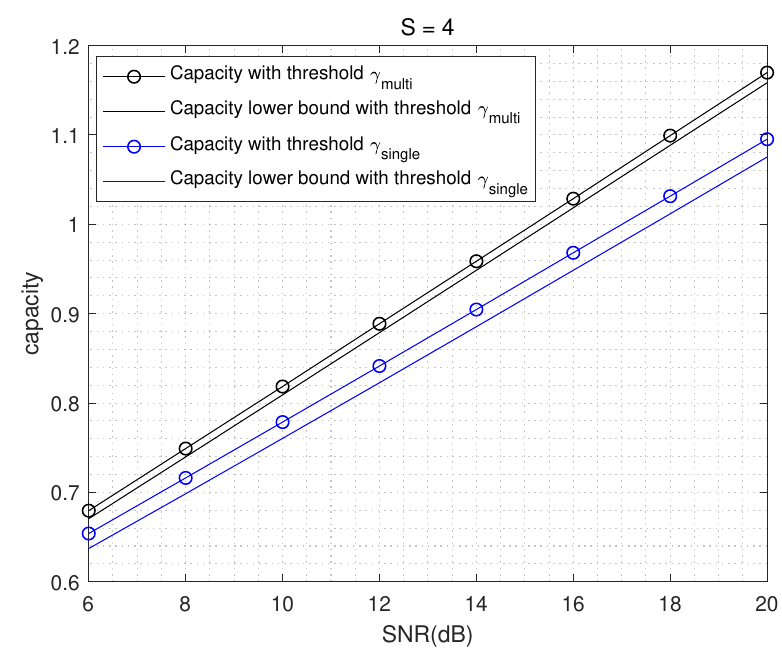}}
		\subfigure[\scriptsize  Comparsion of $\gamma_{multi}$ and $\gamma_{single}$.]{\includegraphics[width=0.4\textwidth, trim=  100 265 120 265,clip]{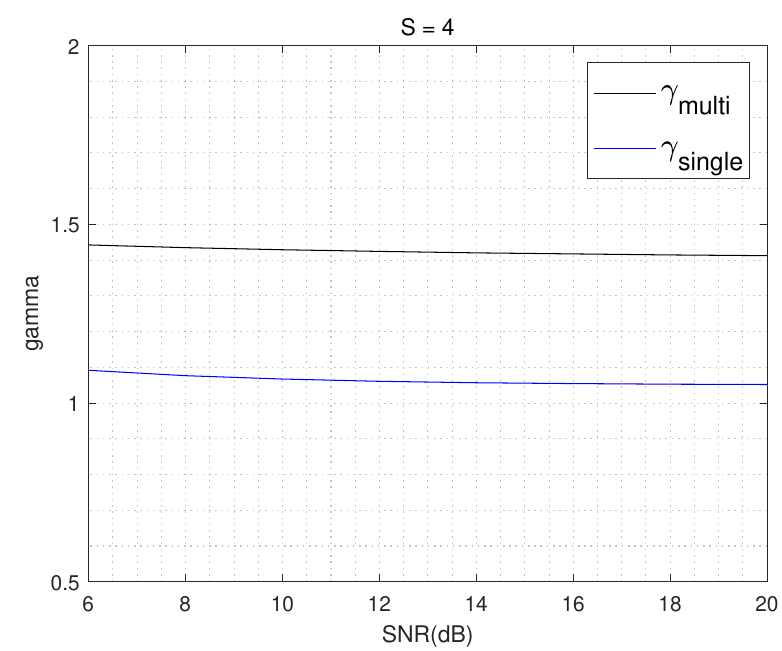}}
	\end{center}
	\caption{Capacity and capacity lower bound under different thresholds in fast fading channel.}\label{fig41}
\end{figure}

\section{Conclusion}
\label{sec:conclusion}
In this paper, we develop analytical models for coded water-filling techniques to evaluate the system-level gains of early stopping in multi-cell, multi-user scenarios. Under the AWGN channel, we demonstrate that the interference cancellation benefits of early stopping originate from the disparity in the number of information bits transmitted by users. In fading channel, even when user parameters are symmetric, interference cancellation gain is still attainable thanks to the stochastic nature of fading coefficients. In addition to performance analysis, we also studied the optimal power allocation strategy for coded water-filling. In particular, we focus on the practical scenario where an individual user only knows its own channel state information. We show that under multi-user conditions, the optimal threshold for water-filling is higher than its single-user threshold. This indicates that pausing transmission at lower SNR to minimize interference can lead to higher system capacity. These theoretical results underscore the significance of coded water-filling in a potentially crowded multi-cell wireless network.


\begin{thebibliography}{10}

          \bibitem{Strass}
V. Strassen, ``Asymptotic estimates in Shannon's information theory,''{
\em Trans. 3rd Prague Conf. Inf. Theory}, Prague, 1962, pp. 689-723.
		\bibitem{poly1}
		Y. Polyanskiy, H. V. Poor and S. Verdu, ``Channel Coding Rate in the Finite Blocklength Regime,'' \textit{ IEEE Transactions on Information Theory}, vol. 56, no. 5, pp. 2307-2359, May 2010.		
		\bibitem{shannon}
		C. E. Shannon, ``The Zero Error Capacity of a Noisy Channel,'' \textit{ IRE Transactions on Information Theory}, vol. 2, no. 3, pp. 8–19, Sep. 1956.
		\bibitem{fast1}
		 M. Horstein, ``Sequential Transmission Using Noiseless Feedback,'' \textit{ IEEE Transactions on Information Theory}, vol. 9, no. 3, pp. 136–143, Jul. 1963.
		\bibitem{fast2}
		 J. Schalkwijk and T. Kailath, ``A Coding Scheme for Additive Noise
		Channels with Feedback-I: No Bandwidth Constraint,'' \textit{ IEEE Transactions on Information Theory}, vol. 12, no. 2, pp. 172–182, Apr. 1966.

		\bibitem{poly2}
		Y. Polyanskiy, H. V. Poor and S. Verdu, ``Feedback in the Non-Asymptotic Regime,''  \textit{ IEEE Transactions on Information Theory}, vol. 57, no. 8, pp. 4903-4925, Aug. 2011.
		\bibitem{Tril1}
		K. F. Trillingsgaard and P. Popovski, ``Variable-Length Coding for Short Packets over a Multiple Access Channel with Feedback,'' \textit{11th International Symposium on Wireless Communications Systems}, Barcelona, Spain, 2014, pp. 796-800.
		\bibitem{gassian1}
		L. V. Truong and V. Y. F. Tan, ``On Gaussian MACs with Variable-Length Feedback and Non-vanishing Error Probabilities,'' \textit{ IEEE Transactions on Information Theory}, vol. 64, no. 4, pp. 2333–2346, Apr. 2018.
		\bibitem{compound}
		 Y. Polyanskiy, ``On Dispersion of Compound DMCs,'' \textit{51st Annual Allerton Conference on Communication, Control, and Computing}, Monticello, IL, USA,
		Oct. 2013, pp. 26–32.
		\bibitem{Tril2}
		K. F. Trillingsgaard, W. Yang, G. Durisi and P. Popovski, ``Broadcasting a Common Message with Variable-Length Stop-Feedback Codes,'' \textit{IEEE International Symposium on Information Theory}, Hong Kong, China, 2015, pp. 2505-2509.
		\bibitem{Tril3}
		K. F. Trillingsgaard, W. Yang, G. Durisi and P. Popovski, ``Common-Message Broadcast Channels with Feedback in the Nonasymptotic Regime: Stop Feedback,'' \textit{ IEEE Transactions on Information Theory}, vol. 64, no. 12, pp. 7686-7718, Dec. 2018.
        \bibitem{Will1}
        A. R. Williamson, T.-Y. Chen, and R. D. Wesel, ``Variable-Length Convolutional Coding for Short Blocklengths with Decision Feedback,''
        \textit{ IEEE Transactions on Communications}, vol. 63, no. 7, pp. 2389–2403, Jul. 2015.
		\bibitem{time1}
		H. Yang, R. C. Yavas, V. Kostina, and R. D. Wesel, ``Variable-Length
		Stop-Feedback Codes with Finite Optimal Decoding Times for BI-AWGN
		Channels,'' \textit{IEEE International Symposium on Information Theory}, Espoo, Finland,
		Jun. 2022, pp. 2327–2332.
		\bibitem{review}
		R. C. Yavas, V. Kostina and M. Effros, ``Variable-Length Sparse Feedback Codes for Point-to-Point, Multiple Access, and Random Access Channels,'' \textit{ IEEE Transactions on Information Theory}, vol. 70, no. 4, pp. 2367-2394, April 2024.
		
		\bibitem{shannon1}
		C. E. Shannon, ``A Mathematical Theory of Communication,'' \textit{ The Bell System Technical Journal},  , vol. 27, no. 3, pp. 379-423, July 1948.
		
		\bibitem{shannon2}
		C. E. Shannon, ``Communication in the Presence of Noise,'' \textit{ Proceedings of the IRE}, vol. 37, no. 1, pp. 10-21, Jan. 1949.
		\bibitem{timeconti1}
		A. D. Wyner, ``The Capacity of the Band-Limited Gaussian Channel,'' \textit{ The Bell System Technical Journal}, vol. 45, no. 3, pp. 359-395, Mar. 1966.
		\bibitem{timeconti2}
		R. G. Gallager, ``Information Theory and Reliable Communication,'' New York: Wiley, 1968.
		\bibitem{timeconti3}
		D. Slepian and H. O. Pollak, ``Prolate Spheroidal Wave Functions, Fourier
		Analysis and Uncertainty: Part I,'' \textit{ The Bell System Technical Journal}, vol. 40, no. 1, pp. 43-63, Jan. 1961.
		\bibitem{timeconti4}
		H. J. Landau and H. O. Pollak, ``Prolate Spheroidal Wave Functions, Fourier
		Analysis and Uncertainty: Part II,'' \textit{ The Bell System Technical Journal}, vol. 40, no. 1, pp. 65-84, Jan. 1961.
		\bibitem{timeconti5}
		H. J. Landau and H. O. Pollak, ``Prolate Spheroidal Wave Functions, Fourier
		Analysis and Uncertainty: Part III,'' \textit{ The Bell System Technical Journal}, vol. 41, no. 4, pp. 1295-1336, Jul. 1962.

		\bibitem{water2}
		T. M. Cover and J. A. Thomas, ``Elements of Information Theory,'' New York: Wiley, 1991.
		
		\bibitem{water3.1}
		G. G. Raleigh and J. M. Cioffi, ``Spatio-Temporal Coding for Wireless Communication,'' \textit{ IEEE Transactions on Communications}, vol. 46, no. 3, pp. 357–366, Mar. 1998.
		
        \bibitem{water3.2}
        A.Scaglione, S. Barbarossa, and G. B. Giannakis, ``Filterbank Transceivers Optimizing Information Rate in Block Transmissions over Dispersive Channels,'' \textit{ IEEE Transactions on Information Theory}, vol. 45, no. 3, pp. 1019–1032, Apr. 1999.

		\bibitem{water3.3}
		A. Scaglione, G. B. Giannakis, and S. Barbarossa, ``Redundant Filterbank Precoders and Equalizers Part I: Unification and Optimal Designs,'' \textit{ IEEE Transactions on Signal Processing}, vol. 47, no. 7, pp. 1988–2006, Jul. 1999.
		
		\bibitem{water4}
		E. N. Onggosanusi, A. M. Sayeed, and B. D. V. Veen, ``Efficient Signaling Schemes for Wideband Space-Time Wireless Channels Using Channel State Information,'' \textit{ IEEE Transactions on Vehicular Technology}, vol. 52, no. 1, pp. 1–13, Jan. 2003.
		
		\bibitem{constant1}
		A. J. Goldsmith and P. P. Varaiya, ``Capacity of Fading Channels with Channel Side Information,'' \textit{ IEEE Transactions on Information Theory}, vol. 43, no. 6, pp. 1986-1992, Nov. 1997.

\bibitem{constant11}
P. S. Chow, ``Bandwidth Optimized Digital Transmission Techniques for Spectrally Shaped Channels with Impulse Noise,'' Ph.D. thesis, Stanford University, 1993.


		\bibitem{constant2}
		Wei Yu and J. M. Cioffi, ``On Constant Power Water-Filling,'' \textit{IEEE International Conference on Communications,} Helsinki, Finland, 2001, pp. 1665-1669.

		\bibitem{water15}
		H. Moon, ``Waterfilling Power Allocation at High SNR Regimes,'' \textit{ IEEE Transactions on Communications}, vol. 59, no. 3, pp. 708-715, Mar. 2011.

		\bibitem{constant3}
		R. Knopp and P. A. Humblet, ``Information Capacity and Power Control in Single-Cell Multiuser Communications,'' \textit{IEEE International Conference on Communications}, Seattle, WA, USA, 1995, pp. 331-335.
		
		\bibitem{water5}
		D. P. Palomar, J. M. Cioffi, and M. A. Lagunas, ``Joint Tx-Rx Beamforming Design for Multicarrier MIMO Channels: A Unified Framework for Convex Optimization,'' \textit{ IEEE Transactions on Signal Processing}, vol. 51, no. 9, pp. 2381–2401, Sep. 2003.
		
		\bibitem{water6.1}
		J. Yang and S. Roy, ``Joint Transmitter-Receiver Optimization for Multi-Input Multi-Output Systems with Decision Feedback,'' \textit{ IEEE Transactions on Information Theory}, vol. 40, no. 5, pp. 1334–1347, Sept. 1994.	

		\bibitem{water13}
		G. Scutari, D. P. Palomar and S. Barbarossa, ``The MIMO Iterative Waterfilling Algorithm,'' \textit{ IEEE Transactions on Signal Processing}, vol. 57, no. 5, pp. 1917-1935, May 2009.
		
		\bibitem{water7}
		J. Xu, L. Qiu and S. Zhang, ``Energy Efficient Iterative Waterfilling for the MIMO Broadcasting Channels,'' \textit{IEEE Wireless Communications and Networking Conference}, Paris, France, 2012, pp. 198-203.
		
		\bibitem{water8}
		D. Park, ``Iterative Waterfilling with User Selection in Gaussian MIMO Broadcast Channels,'' \textit{ IEEE Transactions on Communications}, vol. 66, no. 5, pp. 1902-1911, May 2018.
		
		\bibitem{water9}
		Y. Lu and W. Zhang, ``Water-Filling Capacity Analysis in Large MIMO Systems,'' \textit{ Computing, Communications and IT Applications Conference}, Hong Kong, China, 2013, pp. 186-190.
		
		\bibitem{water10}
		D. P. Palomar and M. A. Lagunas, ``Simplified Joint Transmit-Receive Space-Time Equalization on Spatially Correlated MIMO Channels: A Beamforming Approach,'' \textit{ IEEE Journal on Selected Areas in Communications}, vol. 21, no. 5, pp. 730–743, Jun. 2003.
		
		\bibitem{water11}
		D. P. Palomar, ``A Unified Framework for Communications Through MIMO Channels,'' Ph.D. Dissertation, Techn. Univ. Catalonia (UPC), Barcelona, Spain, May 2003.
		
		\bibitem{water12}
		D. P. Palomar and J. R. Fonollosa, ``Practical Algorithms for a Family of Waterfilling Solutions,'' \textit{ IEEE Transactions on Signal Processing}, vol. 53, no. 2, pp. 686-695, Feb. 2005.
				
		
		\bibitem{water20}
		W. Yu, ``Multiuser Water-Filling in the Presence of Crosstalk,'' \textit{Information Theory and Applications Workshop}, La Jolla, CA, USA, 2007, pp. 414-420.
		
		
		\bibitem{water22}
		B. Luo, Q. Cui, H. Wang and X. Tao, ``Optimal Joint Water-Filling for Coordinated Transmission over Frequency-Selective Fading Channels,'' \textit{IEEE Communications Letters}, vol. 15, no. 2, pp. 190-192, February 2011.
		
		\bibitem{water23}
		Z. Wang, V. Aggarwal and X. Wang, ``Iterative Dynamic Water-Filling for Fading Multiple-Access Channels with Energy Harvesting,'' \textit{ IEEE Journal on Selected Areas in Communications}, vol. 33, no. 3, pp. 382-395, March 2015.
		
		
		\bibitem{water27}
		L. Lai and H. El Gamal, ``The Water-Filling Game in Fading Multiple-Access Channels,'' \textit{ IEEE Transactions on Information Theory}, vol. 54, no. 5, pp. 2110-2122, May 2008.

\bibitem{order_sta}
B. C. Arnold, N. Balakrishnan, and H. N. Nagaraja, ``A First Course in Order Statistics,'' \textit{ Society for Industrial and Applied Mathematics}, 2008.

	\end{thebibliography}
\end{document}